\documentclass[letterpaper, 10 pt, conference]{ieeeconf}  

\IEEEoverridecommandlockouts                              

\usepackage{cite}
\usepackage{enumerate}
\usepackage{graphicx}
\usepackage[cmex10]{amsmath} 
\usepackage{amssymb}
\usepackage{mathtools}
\usepackage{xspace}
\usepackage{subfigure}
\usepackage{multirow}
\usepackage[lined,linesnumbered,ruled,commentsnumbered]{algorithm2e} 
\usepackage{colonequals}
\usepackage[mathscr]{euscript}
\usepackage{mathrsfs}

\usepackage{epsfig}
\usepackage{epstopdf}

\usepackage{tikz}
\usepackage{pgfplots}

\usetikzlibrary{arrows,shapes,backgrounds,plotmarks,positioning}

\pgfplotsset{compat=newest}                         
\pgfplotsset{plot coordinates/math parser=false}
\newlength\figureheight
\newlength\figurewidth

\newtheorem{theorem}{Theorem}
\newtheorem{lemma}{Lemma}
\newtheorem{definition}{Definition}

\newtheorem{proposition}{Proposition}

\newtheorem{example}{Example}

\newcommand{\argmin}{\arg\!\min}

\newcommand{\op}{\text}
\newcommand{\RZ}[1]{\mathsf{Z}_{#1}}

\newcommand{\RW}[1]{\mathsf{W}_{#1}}
\newcommand{\RCO}{R_{\op{CO}}}
\newcommand{\RACO}{R_{\op{CO}}}

\newcommand{\RRCO}{\mathscr{R}}

\newcommand{\D}{\mathcal{L}}

\newcommand{\XComp}{X^*}
\newcommand{\alphaComp}{\alpha^*}

\newcommand{\CoordSatCapFus}{\text{CoordSat}}
\newcommand{\StrMap}{\text{StrMap}}

\newcommand{\Set}[1]{\{#1\}}

\newcommand{\Pat}{\mathcal{P}}
\newcommand{\Qat}[2]{\mathcal{Q}_{#1,#2}}
\newcommand{\X}{\mathcal{X}}
\newcommand{\TX}{\tilde{\mathcal{X}}}
\newcommand{\Y}{\mathcal{Y}}
\newcommand{\TY}{\tilde{\mathcal{Y}}}
\newcommand{\M}{\mathcal{M}}
\newcommand{\TM}{\tilde{\mathcal{M}}}
\newcommand{\N}{\mathcal{N}}
\newcommand{\TN}{\tilde{\mathcal{N}}}

\newcommand{\U}[2]{\mathcal{U}_{#1,#2}}
\newcommand{\TU}[2]{\tilde{\mathcal{U}}_{#1,#2}}

\newcommand{\Up}[1]{\mathcal{S}_{#1}}
\newcommand{\TUp}[1]{\tilde{{\mathcal{S}}}_{#1}}
\newcommand{\SU}{\mathcal{S}}
\newcommand{\TSU}{\tilde{\mathcal{S}}}

\newcommand{\Patp}[1]{\mathcal{P}^{(#1)}}
\newcommand{\alphap}[1]{\alpha^{(#1)}}
\newcommand{\alphaU}{\underline{\alpha}}

\newcommand{\rv}{\mathbf{r}} 

\newcommand{\Fu}[1]{f_{#1}}
\newcommand{\FuHat}[1]{\hat{f}_{#1}}
\newcommand{\FuU}[1]{g_{#1}}

\newcommand{\Real}{\mathbb{R}}
\newcommand{\RealP}{\mathbb{R}_{+}}    

\newcommand{\SFM}{\text{SFM}}

\title{\LARGE \bf Improving Computational Efficiency of Communication for Omniscience and Successive Omniscience }

\author{Ni~Ding$^{1}$, Parastoo~Sadeghi$^{2}$ and Thierry~Rakotoarivelo$^{1}$
\thanks{$^{1}$Ni Ding and Thierry Rakotoarivelo are with Data61 (email: $\{$ni.ding, thierry.rakotoarivelo$\}$@data61.csiro.au).}%
\thanks{$^{2}$Parastoo Sadeghi is with the Research School of Engineering, College of Engineering and Computer Science, the Australian National University (email: $\{$parastoo.sadeghi$\}$@anu.edu.au).}%
}

\begin{document}

\maketitle
\thispagestyle{empty}
\pagestyle{empty}

\begin{abstract}
    For a group of users in $V$ where everyone observes a component of a discrete multiple random source, the process that users exchange data so as to reach omniscience, the state where everyone recovers the entire source, is called communication for omniscience (CO).
    We first consider how to improve the existing complexity $O(|V|^2 \cdot \SFM(|V|)$ of minimizing the sum of communication rates in CO, where $\SFM(|V|)$ denotes the complexity of minimizing a submodular function.
    We reveal some structured property in an existing coordinate saturation algorithm: the resulting rate vector and the corresponding partition of $V$ are segmented in $\alpha$, the estimation of the minimum sum-rate.
    A parametric (PAR) algorithm is then proposed where, instead of a particular $\alpha$, we search the critical points that fully determine the segmented variables for all $\alpha$ so that they converge to the solution to the minimum sum-rate problem and the overall complexity reduces to $O(|V| \cdot \SFM(|V|))$.

    For the successive omniscience (SO), we consider how to attain local omniscience in some complimentary user subset so that the overall sum-rate for the global omniscience still remains minimum. While the existing algorithm only determines a complimentary user subset in $O(|V| \cdot \SFM(|V|))$ time, we show that, if a lower bound on the minimum sum-rate is applied to the segmented variables in the PAR algorithm, not only a complimentary subset, but also an optimal rate vector for attaining the local omniscience in it are returned in $O(|V| \cdot \SFM(|V|))$ time.
\end{abstract}

\section{introduction}
Let there be a finite number of users that are indexed by the set $V$. Each of them observes a distinct component of a discrete multiple random source in private. The users are allowed to exchange their observations over public authenticated noiseless broadcast channels so as to attain \textit{omniscience}, the state that each user reconstructs all components in the multiple source. This process is called \emph{communication for omniscience (CO)} \cite{Csiszar2004}, where the fundamental problem is how to attain omniscience with the minimum sum of broadcast rates.

The study of the minimum sum-rate problem started in the coded cooperative data exchange (CCDE) \cite{Roua2010,CourtIT2014,MiloIT2016}, where the users obtain finite-length observations of a multiple linear source and the broadcasts are an integral number of linear combinations. The authors in \cite{CourtIT2014,MiloIT2016} solved
the minimum sum-rate problem in CCDE in polynomial time by $O(|V|^2)$ calls of the submodular function minimization (SFM) algorithm. In the general models where the observation sequences are infinitely long and communication rates are real-valued, it is shown in \cite{Ding2016NetCod,Ding2018IT,ChanMMI} that the minimum sum-rate is determined by the first critical point in the principal sequence of partitions (PSP) \cite{Narayanan1991PLP,MinAveCost}, a partition sequence that determines a Dilworth truncation parameterized by an estimation $\alpha$ of the minimum sum-rate.
A modified decomposition algorithm (MDA) is proposed in \cite{Ding2016NetCod,Ding2018IT} that recursively uses the Dilworth truncation to update the current estimation $\alpha$ until it converges to the minimum. The asymptotic complexity of the MDA is again $O(|V|^2 \cdot \SFM(|V|))$.\footnote{The asymptotic time complexity refers to the asymptotic measure of the (worst-case) time complexity of an algorithm when the number of users $|V|$ grows large \cite{BigO1965}. $\SFM(|V|)$ denotes the asymptotic complexity of minimizing a submodular function that is defined on $2^V$.}
However, $\SFM(|V|)$ is in the order of $|V|^5$ to $|V|^8$ \cite[Chapter~VI]{Fujishige2005} and yet it is still worth discussing how to reduce this complexity.

In the meantime, the idea of successive omniscience (SO) is proposed in \cite{ChanSuccessiveIT,Ding2015NetCod} revealing that the CO problem can be solved in a two-stage manner: local omniscience followed by the global omniscience. It is shown in \cite{ChanSuccessiveIT} that there is a particular group of \emph{complimentary} user subsets so that the local omniscience in any of them can be attained first while the overall communication rates for the global omniscience whereafter still remains minimized. Then, the SO boils down to the problem of how to select a complimentary subset and to attain omniscience in it. The authors in \cite{Ding2017ISIT} derived a sufficient condition for a user subset to be complimentary so that such a subset can be searched in $O(|V| \cdot \SFM(|V|))$ time. However, we are still missing an optimal rate vector for the local omniscience problem. To avoid repetitively applying the existing algorithms for CO,\footnote{The local omniscience in a user subset is again a CO problem where the existing algorithms, e.g., the MDA in \cite{Ding2016NetCod,Ding2018IT}, can be applied. However, in the case of multi-stage SO, we need to call these algorithms time and again. It makes the SO more complex than attaining omniscience in a one-off manner. }
it is desirable to see if an optimal rate vector for the local omniscience can be obtained at the same time when a complimentary user subset is chosen.

The main purpose of this paper is to improve the computational efficiency for solving both CO and SO problems. We first show how to reduce the complexity of solving the minimum sum-rate problem from $O(|V|^2 \cdot \SFM(|V|))$ to $O(|V| \cdot \SFM(|V|))$. We start the work by reviewing the coordinate saturation ($\CoordSatCapFus$) algorithm,\footnote{In this paper, the $\CoordSatCapFus$ algorithm refers to \cite[Algorithm~3]{Ding2018IT}.} a nesting algorithm in the MDA in \cite{Ding2018IT} that determines the Dilworth truncation for a given $\alpha$. We show that the partition obtained in each iteration gets coarser in $\alpha$ and its value is segmented by a finite number of critical points that also form a PSP. We then consider a question that follows naturally: how to obtain this partition for all, instead of one, $\alpha$.
For the user subset that is used to update the partition,\footnote{In the $\CoordSatCapFus$ algorithm, the partition is updated by merging all elements that intersect with a user subset.}
we prove a strict strong map property showing that its size is shrinking and also segmented in $\alpha$. All critical values for determining this segmented subset can be searched by the parametric submodular function minimization (PSFM) algorithm \cite{Fleischer2003PSFM,Nagano2007PSFM,IwataPSFM1997} that completes at the same asymptotic time as the SFM algorithm.
We then propose a parameterized (PAR) algorithm that iteratively updates a segmented partition and a rate vector towards the PSP of the entire user set, which determines the minimum sum-rate, and an optimal rate vector, respectively.
The complexity of the PAR algorithm reduces to $O(|V| \cdot \SFM(|V|))$. We also discuss how to implement the PAR algorithm in a distributed manner.

In the second part of this paper, we discuss how to solve the SO problem efficiently. We apply a lower bound on the minimum sum-rate to the segmented partition and rate vector at the end of each iteration of the PAR algorithm to show that not only a complimentary user subset for SO, but also an optimal rate vector that attains the local omniscience in it can both be determined in $O(|V| \cdot \SFM(|V|))$ time.

This paper is organized as follows. The system model for CO is described in Section~\ref{sec:CO}, where we also introduce the notation and review the existing results that are required to prove the strong map property and propose the PAR algorithm in Section~\ref{sec:ParAlgo}. In Section~\ref{sec:Complexity}, we show the complexity reduction by the PAR algorithm and discuss a distributed implementation method. In Section~\ref{sec:SO}, we utilize the PAR algorithm to efficiently solve the SO problem.

\section{Communication for Omniscience}
\label{sec:CO}

Let $V = \Set{1,\dotsc,|V|}$ with $|V|>1$ be a finite set that indexes all users in the system. We call $V$ the \textit{ground set}. Let $\RZ{V}=(\RZ{i}:i\in V)$ be a vector of discrete random variables indexed by $V$. For each $i\in V$, user $i$ privately observes an $n$-sequence $\RZ{i}^n$ of the random source $\RZ{i}$ that is i.i.d.\ generated according to the joint distribution $P_{\RZ{V}}$. Each user is able to compress and publish his/her observations over noiseless broadcast channels to help the others reconstruct the source sequence $\RZ{V}^n$. The state that each user recovers $\RZ{V}^n$ is called \textit{omniscience} and the process that the users communicate with each other to attain omniscience is called \textit{communication for omniscience (CO)} \cite{Csiszar2004}.

\subsection{Minimum Sum-rate Problem}

Let $\rv_V=(r_i:i\in V)$ be a rate vector and $r$ be the \textit{sum-rate function} associated with $\rv_V$ such that
$$ r(X)=\sum_{i\in X} r_i, \quad \forall X \subseteq V $$
with the convention $r(\emptyset)=0$. We call $\rv_V$ an \textit{achievable rate vector} if the omniscience can be attained by letting the users communicate at the rates designated by $\rv_V$.
The authors in \cite{Csiszar2004} derived the \textit{achievable rate region} in terms of the multiterminal Slepian-Wolf constraint \cite{SW1973,Cover1975}:
$$ \RRCO(V)=\Set{ \rv_V\in\Real^{|V|} \colon r(X) \geq H(X|V\setminus X),\forall X \subsetneq V },$$
where $H(X)$ is the amount of randomness in $\RZ{X}$ measured by Shannon entropy \cite{Cover2012ITBook} and $H(X|Y)=H(X \cup Y)-H(Y)$ is the conditional entropy of $\RZ{X}$ given $\RZ{Y}$.

The fundamental problem in CO is to minimize the sum-rate in the achievable rate region \cite[Proposition 1]{Csiszar2004}
\begin{equation} \label{eq:MinSumRate}
    \RCO(V) = \min\Set{ r(V) \colon \rv_V \in \RRCO(V)}.
\end{equation}
To efficiently solve the \textit{minimum sum-rate problem} by avoiding dealing with the exponentially growing number of constraints in the linear programming, \eqref{eq:MinSumRate} is converted to a combinatorial optimization problem \cite[Example 4]{Csiszar2004} \cite{Chan2008tight} \cite[Corollary 6]{Ding2018IT}
\begin{equation} \label{eq:MinSumRatePat}
    \RACO(V) = \max_{\Pat \in \Pi(V) \colon |\Pat| > 1} \sum_{C \in \Pat} \frac{H(V) - H(C)}{|\Pat|-1},
\end{equation}
where $\Pi(V)$ denotes the set of all partitions of $V$. It is shown in \cite{ChanMMI,Ding2016NetCod,Ding2018IT} that problem \eqref{eq:MinSumRatePat} can be solved based on the existing submodular function minimization (SFM) techniques in strongly polynomial time $O(|V|^2 \cdot \SFM(|V|))$.

\subsection{Existing Results}
\label{sec:ExResults}

In this section, we introduce the notation and two concepts relating to the minimum sum-rate problem: the Dilworth truncation and the principal sequence of partitions (PSP). We also present the coordinatewise saturation ($\CoordSatCapFus$) algorithm, an essential nesting algorithm in the MDA algorithm in \cite{Ding2018IT} for solving~\eqref{eq:MinSumRatePat}. The purpose is to review the existing results that are required to prove the strong map property in Section~\ref{sec:ParAlgo}.

\subsubsection{Preliminaries}

For $X \subseteq V$, let $\chi_X = (r_i \colon i \in V )$ be the \textit{characteristic vector} of the subset $X$ such that $r_i = 1$ if $i \in X$ and $r_i = 0$ if $i \notin X$. The notation $\chi_{\Set{i}}$ is simplified by $\chi_i$.
For $i \in V$, let $V_i = \Set{1,\dotsc,i}$ be the set of the first $i$ users in $V$.
Let $\sqcup$ denote the disjoint union. For $\X$ that contains disjoint subsets of $V$, we denote by $\TX = \sqcup_{C \in \X} C$ the \textit{fusion} of $\X$. For example, for $\X = \Set{\Set{3,4},\Set{2},\Set{8}}$, $\TX = \Set{2,3,4,8}$.
For partitions $\Pat,\Pat' \in \Pi(V)$, we denote by $\Pat \preceq \Pat'$ if $\Pat$ is finer than $\Pat'$ and $\Pat \prec \Pat'$ if $\Pat$ is strictly finer than $\Pat'$.

Function $f \colon 2^V \mapsto \Real$ is \textit{submodular} if $f(X) + f(Y) \geq f(X \cap Y) + f(X \cup Y)$ for all $X,Y \subseteq V$. The problem $\min \Set{f(X) \colon X \subseteq V}$ is a submodular function minimization (SFM) problem. It can be solved in strongly polynomial time and the set of minimizers $\argmin \Set{f(X) \colon X \subseteq V}$ form a set lattice such that the smallest minimizer $\bigcap \argmin \Set{f(X) \colon X \subseteq V}$ and largest minimizer $\bigcup \argmin \Set{f(X) \colon X \subseteq V}$ uniquely exist and can be determined at the same time when the SFM problem is solved \cite[Chapter~VI]{Fujishige2005}.
$P(f,\leq) = \Set{\rv_V \colon r(X) \leq f(X), X \subseteq V}$ and $B(f,\leq) = \Set{\rv_V \in P(f,\leq) \colon r(V) = f(V)}$ are the \emph{submodular polyhedron} and \emph{base polyhedron}, respectively.
We call $f^{V_i} \colon 2^{V_i} \mapsto \Real$ such that $f^{V_i}(X) = f(X)$ for all $X \subseteq V_i$ the \textit{reduction} of $f$ on $V_i$ \cite[Section~3.1(a)]{Fujishige2005}.

\subsubsection{Dilworth Truncation}

For $\alpha \in \RealP$, define a set function $\Fu{\alpha} \colon 2^{V} \mapsto \Real$ such that $\Fu{\alpha}(X) = \alpha - H(V) +H(X), \forall X \subseteq V$ except that $f(\emptyset) = 0$. Let $\Fu{\alpha}[\cdot]$ be a partition function such that $\Fu{\alpha}[\Pat] = \sum_{C \in \Pat} \Fu{\alpha}(C)$ for all $\Pat \in \Pi(V)$. The Dilworth truncation of $\Fu{\alpha}$ is \cite{Dilworth1944}
\begin{equation} \label{eq:Dilworth}
    \FuHat{\alpha}(V) =  \min_{\Pat \in \Pi(V)} \Fu{\alpha}[\Pat].
\end{equation}
Let
$ \Qat{\alpha}{V} = \bigwedge \argmin_{\Pat \in \Pi(V)} \Fu{\alpha}[\Pat] $
denote the finest minimizer of \eqref{eq:Dilworth}.\footnote{The minimizers of \eqref{eq:Dilworth} form a partition lattice such that the finest and coarsest minimizers uniquely exist \cite{Narayanan1991PLP}.}
The solution to \eqref{eq:Dilworth} exhibits a strong structure in $\alpha$ that is described by the PSP.

\subsubsection{Principal Sequence of Partitions (PSP)}
\label{subsec:PSP}

The Dilworth truncation $\FuHat{\alpha}(V)$ is piecewise linear strictly increasing in $\alpha$. It is determined by $p < |V| $ critical points $\alphap{p} < \dotsc < \alphap{1} < \alphap{0} = H(V)$
with the corresponding finest minimizer $\Patp{j} = \Qat{\alphap{j}}{V} = \bigwedge \argmin_{\Pat \in \Pi(V)} \Fu{\alpha_j}[\Pat]$ for all $j \in \Set{0,\dotsc,p}$ forming a partition chain called the \textit{PSP} of the ground set $V$
$$ \Set{\Set{i} \colon i\in V} = \Patp{p} \prec \dotsc \prec \Patp{1} \prec \Patp{0} = \Set{V} $$
such that $\Qat{\alpha}{V} = \Patp{j}$ for all $\alpha \in (\alphap{j+1}, \alphap{j}] $ \cite{MinAveCost,Narayanan1991PLP}.
The first critical point of PSP provides the solutions to the minimum sum-rate problem \cite[Corollary A.3]{Ding2018IT}: $\RCO(V) = \alphap{1}$ and $\Patp{1}$ is the finest maximizer of \eqref{eq:MinSumRatePat}.

\subsubsection{$\CoordSatCapFus$ Algorithm}

In the MDA algorithm proposed in \cite{Ding2016NetCod,Ding2018IT}, the main task is to solve the minimization problem \eqref{eq:Dilworth} for a given value of $\alpha$ by calling the $\CoordSatCapFus$ algorithm in Algorithm~\ref{algo:CoordSatCapFus}.
In step~\ref{step:MinFus}, the set function $\FuU{\alpha}$ is defined as
\begin{equation} \label{eq:FusFunc}
    \FuU{\alpha}(\TX) = \Fu{\alpha}(\TX) - r_{\alpha}(\TX), \quad \forall \X \subseteq \Qat{\alpha}{V_i},
\end{equation}
where $r_\alpha (X) = \sum_{i \in X} r_{\alpha,i} $ is the sum-rate function of the rate vector $\rv_{\alpha,V} = (r_{\alpha,i} \colon i \in V)$ for a given $\alpha$.

The $\CoordSatCapFus$ algorithm is based on the min-max relationship \cite[Section~2.3]{Fujishige2005}\cite[Lemma~23]{Ding2018IT}
    \begin{subequations}
        \begin{align}
            &\min\Set{ \FuU{\alpha}(\TX) \colon \Set{i} \in \X \subseteq \Qat{\alpha}{V_i}} \label{eq:Fusion} \\
            & \qquad\qquad\qquad  = \max \Set{\xi \colon \rv_V + \xi \chi_i \in P(\Fu{\alpha},\leq)}.  \label{eq:MinMax}
        \end{align}
    \end{subequations}
where \eqref{eq:Fusion} is a SFM problem \cite[Section~V-B]{Ding2018IT} and the maximum of \eqref{eq:MinMax} is called the \emph{saturation capacity}. The idea is to saturate each dimension $i$ of a rate vector $\rv_{\alpha,V} \in P(\Fu{\alpha},\leq)$ until it reaches the base polyhedron $B(\FuHat{\alpha},\leq)$ with $r_{\alpha}(V) = \FuHat{\alpha}(V)$ and $\Qat{\alpha}{V}$ being updated to the finest minimizer of $\min_{\Pat\in \Pi(V)} \Fu{\alpha}[\Pat]$ \cite[Section~V-B]{Ding2018IT}.
Note, we use the notation $\Qat{\alpha}{V_i}$ in Algorithm~\ref{algo:CoordSatCapFus} since we will show in Section ~\ref{subsec:PrePar} that $\Qat{\alpha}{V_i} = \bigwedge \argmin_{\Pat \in \Pi(V_i)} \Fu{\alpha}[\Pat]$ after step~\ref{step:Updates} for all $i$.

Based on the outputs of the $\CoordSatCapFus$ algorithm and the properties of the PSP in Lemma~\ref{lemma:AlphaAdapt} in Appendix~\ref{app:AlphaAdapt}, the MDA algorithm proposed in \cite{Ding2018IT} updates $\alpha$, the estimation of the minimum sum-rate $\RCO(V)$, towards the optimal one and finally returns an optimal rate vector $\rv_{\RCO(V),V} \in B(\FuHat{\RCO(V)},\leq)$ and a partition $\Qat{\RCO(V)}{V}$ being the finest maximizer of \eqref{eq:MinSumRatePat} \cite[Theorem 17]{Ding2018IT}.

        \begin{algorithm} [t]
	       \label{algo:CoordSatCapFus}
	       \small
	       \SetAlgoLined
	       \SetKwInOut{Input}{input}\SetKwInOut{Output}{output}
	       \SetKwFor{For}{for}{do}{endfor}
            \SetKwRepeat{Repeat}{repeat}{until}
            \SetKwIF{If}{ElseIf}{Else}{if}{then}{else if}{else}{endif}
	       \BlankLine
           \Input{$\alpha$, $V$ and $H$}
	       \Output{$\rv_{\alpha,V} \in B(\FuHat{\alpha},\leq)$ and $\Qat{\alpha}{V} = \bigwedge \argmin_{\Pat\in \Pi(V)} \Fu{\alpha}[\Pat]$ }
	       \BlankLine
            Let $\rv_{\alpha,V} \leftarrow (\alpha - H(V)) \chi_V$ so that $\rv_{\alpha,V} \in P(\Fu{\alpha},\leq)$\;
            Initiate $ r_{\alpha,1} \leftarrow \Fu{\alpha}(\Set{1})$ and $\Qat{\alpha}{V_1} \leftarrow \Set{\Set{1}}$\;
            \For{$i=2$ \emph{\KwTo} $|V|$}{
                $\Qat{\alpha}{V_i} \leftarrow \Qat{\alpha}{V_{i-1}} \sqcup \Set{\Set{i}}$ \label{step:PatIni} \;
                $\U{\alpha}{V_i} \leftarrow \bigcap \argmin\Set{ \FuU{\alpha}(\TX) \colon \Set{i} \in \X \subseteq \Qat{\alpha}{V_i}}$\; \label{step:MinFus}
                Update $\rv_{\alpha,V}$ and $\Qat{\alpha}{V_i}$: \label{step:Updates}
                \begin{equation}
                    \begin{aligned}
                        \rv_{\alpha,V_i} &\leftarrow \rv_{\alpha,V_i} + \FuU{\alpha}(\TU{\alpha}{V_i}) \chi_{i}; \\
                        \Qat{\alpha}{V_i} &\leftarrow (\Qat{\alpha}{V_i} \setminus \U{\alpha}{V_i}) \sqcup \Set{ \TU{\alpha}{V_i} };
                    \end{aligned} \nonumber
                \end{equation}
            }
            return $\rv_{\alpha,V}$ and $\Qat{\alpha}{V}$\;
	   \caption{$\CoordSatCapFus$ Algorithm \cite[Algorithm~3]{Ding2018IT}}
	   \end{algorithm}

\section{Parametric Approach}
\label{sec:ParAlgo}

In this section, we reveal the structured properties of the partition $\Qat{\alpha}{V_i}$ and rate vector $\rv_{\alpha,V}$ in the $\CoordSatCapFus$ algorithm and propose a parametric (PAR) algorithm showing that, instead of running the $\CoordSatCapFus$ algorithm for a particular value of $\alpha$, we can obtain $\Qat{\alpha}{V_i}$ and $\rv_{\alpha,V}$ for all $\alpha$ in each iteration $i$.\footnote{In this paper, when we say for all $\alpha$, we mean for all $\alpha \in [0, H(V)]$ since the minimum sum-rate must be in the range $[0, H(V)]$.}
We will show in Section~\ref{sec:Complexity} that this PAR algorithm leads to a complexity reduction.

\subsection{Observations}
\label{subsec:PrePar}

For each $i$, we observe the values of $\Qat{\alpha}{V_i}$ and $\rv_{\alpha,V_i}$ in $\alpha$ in the $\CoordSatCapFus$ algorithm and have the following result.

\begin{proposition} \label{prop:preamble}
    In each iteration $i$ of Algorithm~\ref{algo:CoordSatCapFus}, $\Qat{\alpha}{V_i} = \bigwedge \argmin_{\Pat \in \Pi(V_i)} \Fu{\alpha}[\Pat]$ and $\rv_{\alpha,V_i} \in B(\FuHat{\alpha}^{V_i},\leq)$ for all $\alpha$ after step~\ref{step:Updates}.
\end{proposition}
\begin{proof}
    The proof is straightforward for that $\Qat{\alpha}{V_i} = \bigwedge \argmin_{\Pat \in \Pi(V_i)} \Fu{\alpha}[\Pat]$ and $\rv_{\alpha,V_i} \in B(\FuHat{\alpha}^{V_i},\leq)$ are returned by the call $\CoordSatCapFus(\alpha,V_i,H)$.
\end{proof}

Then, $\Qat{\alpha}{V_i}$ for all $\alpha$ is again characterized by the PSP of $V_i$ with the number of critical points bounded by $|V_i|$.
The function $\FuU{\alpha}$ in \eqref{eq:FusFunc} is defined on $2^{\Qat{\alpha}{V_i}}$ where $\Qat{\alpha}{V_i}$ is a segmented partition variable in $\alpha$ and so is $\rv_{\alpha,V_i}$.

\begin{example} \label{ex:main}
    Consider a $5$-user system with
    \begin{equation}
        \begin{aligned}
            \RZ{1} & = (\RW{a},\RW{b},\RW{c},\RW{d},\RW{f},\RW{g},\RW{i},\RW{j}),  \\
            \RZ{2} & = (\RW{a},\RW{b},\RW{c},\RW{f},\RW{i},\RW{j}),\\
            \RZ{3} & = (\RW{e},\RW{f},\RW{h},\RW{i}),   \\
            \RZ{4} & = (\RW{b},\RW{c},\RW{e},\RW{j}), \\
            \RZ{5} & = (\RW{b},\RW{c},\RW{d},\RW{h},\RW{i}),
        \end{aligned}  \nonumber
    \end{equation}
    where each $\RW{j}$ is an independent uniformly distributed random bit.

    We call $\CoordSatCapFus(\alpha,V,H)$ by setting $\alpha = 3$. We initiate $\rv_{3,V} \leftarrow (\alpha - H(V)) \chi_V = (-7,\dotsc,-7)$, update $r_{3,1} = \Fu{3}(\Set{1}) = 1$ and assign $\Qat{3}{V_1} = \Set{\Set{1}}$. For $i  = 2$, we set $\Qat{3}{V_2} = \Set{\Set{1},\Set{2}}$ and consider the minimization problem $\min\Set{ \FuU{3}(\TX) \colon \Set{2} \in \X \subseteq \Qat{3}{V_2}}$. The minimal minimizer is $\U{3}{V_2} = \Set{\Set{2}}$. We do the updates $r_{3,2} = -7 + \FuU{3}(\TU{3}{V_2}) = -1$ and $\Qat{3}{V_2} = \big( \Set{\Set{1},\Set{2}} \setminus \U{3}{V_2} \big) \sqcup \Set{\TU{3}{V_2}} = \Set{\Set{1},\Set{2}}$. We stop here to consider another value of $\alpha$.

    We call $\CoordSatCapFus(\alpha,V,H)$ by setting $\alpha = 6$. We initiate $\rv_{6,V} \leftarrow (\alpha - H(V)) \chi_V = (-4,\dotsc,-4)$ and set $r_{6,1} = \Fu{6}(\Set{1}) = 4$ and $\Qat{6}{V_1} = \Set{\Set{1}}$. We have $\U{6}{V_2} = \Set{\Set{1},\Set{2}} = \bigcap \argmin\Set{ \FuU{6}(\TX) \colon \Set{2} \in \X \subseteq \Qat{6}{V_2}}$ and do the updates $r_{6,2} = -4 + \FuU{6}(\TU{6}{V_2}) = 0$ and $\Qat{6}{V_2} = \big( \Set{\Set{1},\Set{2}} \setminus \U{6}{V_2} \big) \sqcup \Set{\TU{6}{V_2}} = \Set{\Set{1,2}}$.

    One can show that $\Set{\Set{1},\Set{2}} = \bigwedge \argmin_{\Pat \in \Pi(V_2)} \Fu{3}[\Pat] $ and $\Set{\Set{1,2}} = \bigwedge \argmin_{\Pat \in \Pi(V_2)} \Fu{6}[\Pat] $. In fact, repeating the above procedure for all $\alpha$, we have the segmented $\rv_{\alpha,\Set{1,2}}$ and $\Qat{\alpha}{V_2}$
    \begin{equation} \label{eq:ExQatR}
        \begin{aligned}
            & \rv_{\alpha,V_2} = \begin{cases}
                                    (\alpha-2, \alpha-4) & \alpha \in [0,4], \\
                                    (\alpha-2,0) & \alpha \in (4,10],
                                \end{cases} \\
            & \Qat{\alpha}{V_2} = \begin{cases}
                                    \Set{\Set{1},\Set{2}} & \alpha \in [0,4], \\
                                    \Set{\Set{1,2}} & \alpha \in (4,10],
                                \end{cases}
        \end{aligned}
    \end{equation}
    because of the segmented $\TU{\alpha}{V_2}$
    \begin{equation} \label{eq:ExU}
        \TU{\alpha}{V_2} = \begin{cases}
                             \Set{2} & \alpha \in [0,4],\\
                             \Set{1,2} & \alpha \in (4,10].
                         \end{cases}
    \end{equation}
\end{example}

Proposition~\ref{prop:preamble} suggests that, in each iteration of the $\CoordSatCapFus$ algorithm, we can obtain $\rv_{\alpha,V_i}$ and $\Qat{\alpha}{V_i}$ for all values of $\alpha$, which in fact completes the task of determining the PSP of $V_i$ (see Section~\ref{subsec:PSPExpand}). To do so, it is essential to discuss how to determine $\TU{\alpha}{i}$ for all $\alpha$.
It should also be noted that we automatically know $\U{\alpha}{V_i}$ if $\TU{\alpha}{V_i}$ is obtained in that $\U{\alpha}{V_i} = \Set{ C \in \Qat{\alpha}{V_{i}} \colon C \subseteq \TU{\alpha}{V_i}}$.\footnote{This means that $\U{\alpha}{V_i}$ is the decomposition of $\TU{\alpha}{V_i}$ by $\Qat{\alpha}{V_{i}}$. Here, we should use the value of $\Qat{\alpha}{V_i}$ in the minimization problem $\min\Set{ \FuU{\alpha}(\TX) \colon \Set{i} \in \X \subseteq \Qat{\alpha}{V_i}}$ before the updates in step~\ref{step:Updates}. }
For example, for $\TU{\alpha}{V_2}$ in \eqref{eq:ExU},
    \begin{equation}
        \U{\alpha}{V_2} = \begin{cases}
                             \Set{\Set{2}} & \alpha \in [0,4],\\
                             \Set{\Set{1},\Set{2}} & \alpha \in (4,10].
                         \end{cases}\nonumber
    \end{equation}

\subsection{Strong Map Property}

If $\Qat{\alpha}{V_i} = \bigwedge \argmin_{\Pat \in \Pi(V_i)} \Fu{\alpha}[\Pat]$ after the update $\Qat{\alpha}{V_i} \leftarrow (\Qat{\alpha}{V_i} \setminus \U{\alpha}{V_i}) \sqcup \Set{ \TU{\alpha}{V_i} }$ in step~\ref{step:Updates} in Algorithm~\ref{algo:CoordSatCapFus} based on Proposition~\ref{prop:preamble}, $\Qat{\alpha}{V_i}$ must satisfy the properties of the PSP in Section~\ref{subsec:PSP}: $\Qat{\alpha}{V_i}$ gets monotonically coarser in $\alpha$ and is characterized by a finite number of critical points.\footnote{$\Qat{\alpha}{V_i}$ is in fact the PSP of $V_i$ with an offset in $\alpha$. See Section~\ref{subsec:PSPExpand}.}
It necessarily means the we must have $\TU{\alpha}{V_i}$ being segmented and the size of $\TU{\alpha}{V_i}$ must increase in $\alpha$.
This observations can be justified by the strong map property of the function $\FuU{\alpha}$, which also states that all critical points that characterize the segmented $\TU{\alpha}{V_i}$ can be determined by the parametric submodular function minimization (PSFM) algorithms.

\begin{definition}[strong map {\cite[Section~4.1]{Fujishige2009PP}}]
    For two distributive lattices $\D_1,\D_2 \subseteq 2^V$,\footnote{A group of sets $\D$ form a distributive lattice if, for all $X,Y \in \D$, $X \cap Y \in \D$ and $X \cup Y \in \D$ \cite[Section~3.2]{Fujishige2005}.}
    and submodular functions $\Fu{1} \colon \D_1 \mapsto \Real$ and $\Fu{2} \colon \D_2 \mapsto \Real$, $\Fu{1}$ and $\Fu{2}$ form a strong map, denoted by $\Fu{1} \rightarrow \Fu{2}$, if
    \begin{equation}\label{eq:StrongMap}
        \Fu{1}(Y) - \Fu{1}(X) \geq \Fu{2}(Y) - \Fu{2}(X)
    \end{equation}
    for all $X,Y \in \D_1 \cap \D_2$ such that $X \subseteq Y$. The strong map is strict, denoted by $ \Fu{1} \twoheadrightarrow \Fu{2}$, if $\Fu{1}(Y) - \Fu{1}(X) > \Fu{2}(Y) - \Fu{2}(X)$ for all $X \subsetneq Y$.
\end{definition}

\begin{theorem} \label{theo:StrongMap}
    In each iteration $i$ of the $\CoordSatCapFus$ algorithm, $\FuU{\alpha}$ forms a \textbf{strict strong map} in $\alpha$:
        $$ \FuU{\alpha} \twoheadrightarrow \FuU{\alpha'}, \quad \forall \alpha, \alpha' \colon \alpha < \alpha'. $$
\end{theorem}
\begin{proof}
    For any $\X \subseteq \Qat{\alpha}{V_i}$ and $\Y \subseteq \Qat{\alpha'}{V_i}$ such that $\Set{i} \in \TX \subseteq \TY$ and $\FuU{\alpha}$ and $\FuU{\alpha'}$ are both defined on $\TX$ and $\TY$, we have $i \notin \TY \setminus \TX$. Also, there exist $\M \subseteq \Qat{\alpha}{V_i}$ and $\N \subseteq \Qat{\alpha'}{V_i}$ such that $\TM = \TN = \TY \setminus \TX$ (with $\M \preceq \N$), $r_{\alpha}(\TY \setminus \TX) = \Fu{\alpha}[\M] = \FuHat{\alpha}(\TY \setminus \TX)$ and $r_{\alpha'}(\TY \setminus \TX) = \Fu{\alpha'}[\N] = \FuHat{\alpha'}(\TY \setminus \TX)$.\footnote{The following holds for all $\alpha$ and $i$: (a) $\FuHat{\alpha}(\TX) = \FuHat{\alpha}[X]$ for all $\X \subseteq \Qat{\alpha}{V_i}$ \cite[Theorem~38 and Lemma~39]{Ding2018IT}; (b) $\FuHat{\alpha}(C) = \Fu{\alpha}(C) = r_{\alpha}(C)$ for all $C \in \Qat{\alpha}{V_i}$ \cite[proof of Theorem~38]{Ding2018IT}.}
    Then,
    \begin{equation}
        \begin{aligned}
            & \FuU{\alpha}(\TY) - \FuU{\alpha}(\TX) - \FuU{\alpha'}(\TY) + \FuU{\alpha'}(\TX) \\
            &\qquad\qquad\qquad = r_{\alpha'}(\TY \setminus \TX) - r_{\alpha}(\TY \setminus \TX)\\
            &\qquad\qquad\qquad = \begin{cases}
                    0 & \TX = \TY, \\
                    \FuHat{\alpha'} (\TY \setminus \TX)  - \FuHat{\alpha} (\TY \setminus \TX) & \TX \subsetneq \TY,
                \end{cases}
        \end{aligned} \nonumber
    \end{equation}
    where $\FuHat{\alpha'} (\TY \setminus \TX)  - \FuHat{\alpha} (\TY \setminus \TX) > 0$ for all $\alpha$ and $\alpha'$ such that $\alpha < \alpha'$ since $\FuHat{\alpha}(\TY \setminus \TX)$ is strictly increasing in $\alpha$ (see Section~\ref{subsec:PSP}).
\end{proof}

The strict strong map property directly leads to the structured property of $\TU{\alpha}{V_i}$ in $\alpha$ based on the results in \cite{Fujishige2009PP}.

\begin{lemma}{\cite[Theorems~26 to 28]{Fujishige2009PP}} \label{lemma:PP}
    In each iteration of the $\CoordSatCapFus$ algorithm, the minimal minimizer $\U{\alpha}{V_i}$ of $\min\Set{ \FuU{\alpha}(\TX) \colon \Set{i} \in \X \subseteq \Qat{\alpha}{V_i}}$ satisfies $\TU{\alpha}{V_i} \subseteq \TU{\alpha'}{V_i}$ for all $\alpha < \alpha'$. In addition, $\TU{\alpha}{V_i}$ for all $\alpha$ is fully characterized by $q < |V_i| - 1$ critical points
    $ 0 \leq \alpha_q < \dotsc < \alpha_1 < \alpha_0 = H(V) $ and the corresponding $\TUp{j} = \TU{\alpha_j}{V_i}$ for all $j \in \Set{0,\dotsc,q}$ forms a set chain
    $$ \Set{i} = \TUp{q} \subsetneq  \dotsc \subsetneq \TUp{1} \subsetneq \TUp{0} = V_i $$
    such that $\TU{\alpha}{V_i} = \TUp{q} = \Set{i}$ for all $\alpha \in [0,\alpha_q]$ and $\TU{\alpha}{V_i} = \TUp{j}$ for all $\alpha \in (\alpha_{j+1}, \alpha_j]$ such that $j \in \Set{0,\dotsc,q-1}$.\footnote{It should be noted that the value of $\alphap{j}$s in the PSP and $\alpha_{j}$s in Lemma~\ref{lemma:PP} do not necessarily coincide and the critical points $\alpha_j$s are different for $\min\Set{ \FuU{\alpha}(\TX) \colon \Set{i} \in \X \subseteq \Qat{\alpha}{V_i}}$ with a different $i$. } \hfill\QED
\end{lemma}

\begin{example} \label{ex:PP}
    In Example~\ref{ex:main}, we have $\TU{\alpha}{V_2}$ in \eqref{eq:ExU} characterized by the critical points $\alpha_1 = 4$ and $\alpha_0 = H(V) = 10$ with $\TUp{1} = \Set{2}$ and $\TUp{0} = \Set{1,2}$ such that
    $ \Set{2} = \TUp{1} \subsetneq  \TUp{0} = V_2 $,
    $\TU{\alpha}{V_2} = \TUp{1}$ for $\alpha \in [0,\alpha_1]$ and $\TU{\alpha}{V_2} = \TUp{0}$ for $\alpha \in (\alpha_1,\alpha_0]$.

    We continue the procedure in Example~\ref{ex:main} for $i = 3$ by considering the problem $\min\Set{ \FuU{\alpha}(\TX) \colon \Set{3} \in \X \subseteq \Qat{\alpha}{V_2} \sqcup \Set{\Set{3}}}$ where $\Qat{\alpha}{V_2}$ and $\rv_{\alpha,V_2}$ are in \eqref{eq:ExQatR}. We have
    \begin{equation}  \label{eq:UpV3}
        \TU{\alpha}{V_3} = \begin{cases}
                             \Set{3} & \alpha \in [0,8],\\
                             \Set{1,2,3} & \alpha \in (8,10]
                         \end{cases}
    \end{equation}
    that is determined by the critical points $\alpha_1 = 8$ and $\alpha_0 = H(V) = 10$ with $\TUp{1} = \Set{3}$ and $\TUp{0} = \Set{1,2,3}$ such that
    $ \Set{3} = \TUp{1} \subsetneq  \TUp{0} = V_3 $.
    After the updates in step~\ref{step:Updates}, we have
        \begin{equation}  \label{eq:UpdateV3}
        \begin{aligned}
            & \rv_{\alpha,V_3} = \begin{cases}
                                    (\alpha-2, \alpha-4, \alpha-6) & \alpha \in [0,4], \\
                                    (\alpha-2, 0, \alpha-6) & \alpha \in (4,8], \\
                                    (\alpha-2,0,2) & \alpha \in (8,10],
                                \end{cases} \\
            & \Qat{\alpha}{V_3} = \begin{cases}
                                    \Set{\Set{1},\Set{2},\Set{3}} & \alpha \in [0,4], \\
                                    \Set{\Set{1,2},\Set{3}} & \alpha \in (4,8], \\
                                    \Set{\Set{1,2,3}} & \alpha \in (8,10].
                                \end{cases}
        \end{aligned}
    \end{equation}
\end{example}

\subsection{Parametric Method}

Lemma~\ref{lemma:PP} directly suggests the PAR algorithm in Algorithm~\ref{algo:ParAlgo}, where, in each iteration $i$, we determine the values of $\Qat{\alpha}{V_i}$ and $\rv_{\alpha,V_i}$ for all $\alpha$.\footnote{We call Algorithm~\ref{algo:ParAlgo} a parametric algorithm since the variables $\Qat{\alpha}{V_i}$ and $\rv_{\alpha,V_i}$ in each iteration $i$ are parameterized by $\alpha$ and $\TU{\alpha}{V_i}$ for all $\alpha$ can be determined by a PSFM algorithm that is parameterized by $\alpha$.}

      \begin{algorithm} [t]
	       \label{algo:ParAlgo}
	       \small
	       \SetAlgoLined
	       \SetKwInOut{Input}{input}\SetKwInOut{Output}{output}
	       \SetKwFor{For}{for}{do}{endfor}
            \SetKwRepeat{Repeat}{repeat}{until}
            \SetKwIF{If}{ElseIf}{Else}{if}{then}{else if}{else}{endif}
	       \BlankLine
           \Input{$V$ and $H$}
	       \Output{segmented variables $\rv_V \in B(\FuHat{\alpha},\leq)$ and $\Qat{\alpha}{V} = \bigwedge \argmin_{\Pat\in \Pi(V)} \Fu{\alpha}[\Pat]$ for all $\alpha$}
	       \BlankLine
            $\rv_{\alpha,V} \leftarrow (\alpha - H(V)) \chi_V$ for all $\alpha$\;
            $ r_{\alpha,1} \leftarrow \Fu{\alpha}(\Set{1})$ and $\Qat{\alpha}{V_1} \leftarrow \Set{\Set{1}}$ for all $\alpha$\;
            \For{$i=2$ \emph{\KwTo} $|V|$}{
                $\Qat{\alpha}{V_i} \leftarrow \Qat{\alpha}{V_{i-1}} \sqcup \Set{\Set{i}}$ for all $\alpha$\;
                Obtain the critical points $\Set{\alpha_j \colon j \in \Set{0,\dotsc,q}}$ and $\Set{\TUp{j} \colon j \in \Set{0,\dotsc,q}}$ that determine the minimal minimizer $\U{\alpha}{V_i}$ of $\min\Set{ \FuU{\alpha}(\TX) \colon \Set{i} \in \X \subseteq \Qat{\alpha}{V_i}}$ for all $\alpha$, e.g., by the StrMap algorithm in Algorithm~\ref{algo:PP}\;
                Let $\Gamma_j \leftarrow (\alpha_{j+1},\alpha_{j}]$ for all $j \in \Set{0,\dotsc,q-1}$ and $\Gamma_q \leftarrow [0,\alpha_p]$ and update $\rv_V$ and $\Qat{\alpha}{V_i}$ by \label{step:UpdatesPar}
                \begin{equation}
                    \begin{aligned}
                        \rv_{\alpha,V_i} &\leftarrow \rv_{\alpha,V_i} + \FuU{\alpha}(\TUp{j}) \chi_{i}; \\
                        \Qat{\alpha}{V_i} &\leftarrow (\Qat{\alpha}{V_i} \setminus \Up{j}) \sqcup \Set{ \TUp{j} };
                    \end{aligned} \nonumber
                \end{equation}
                for all $\alpha \in \Gamma_j$;
            }
            return $\rv_V$ and $\Qat{\alpha}{V}$ for all $\alpha$\;
	   \caption{Parametric (PAR) Algorithm}
	   \end{algorithm}

\begin{example} \label{ex:ParAlgo}
    We apply the PAR algorithm to the system in Example~\ref{ex:main}. We first initiate $r_{\alpha,i} = \alpha - H(V)) = \alpha - 10$ for all $i$ and $\alpha$. For $i=1$, we have $\Qat{\alpha}{V_1} = \Set{\Set{1}}$ and $r_{\alpha,1}=\Fu{\alpha}(\Set{1}) = \alpha - 2$ for all $\alpha$. See Fig.~\ref{fig:PSP}(a).

    As shown in Example~\ref{ex:PP}, for $i = 2$ and $i = 3$, we have $\TU{\alpha}{V_2}$ in \eqref{eq:ExU} so that the updated $\rv_{\alpha,V_2}$ and $\Qat{\alpha}{V_2}$ are in \eqref{eq:ExQatR} and $\TU{\alpha}{V_3}$ in \eqref{eq:UpV3} so that the updated $\rv_{\alpha,V_3}$ and $\Qat{\alpha}{V_3}$ are in \eqref{eq:UpdateV3}, respectively. See (b) and (c) in Fig.~\ref{fig:PSP}.

    For $i = 4$, consider the problem $\min\Set{ \FuU{\alpha}(\TX) \colon \Set{4} \in \X \subseteq \Qat{\alpha}{V_4} }$ where $\Qat{\alpha}{V_4} = \Qat{\alpha}{V_3} \sqcup \Set{\Set{4}}$. We have the critical points $\alpha_1 = 7$ and $\alpha_0 = H(V) = 10$ with $\TUp{1} = \Set{4}$ and $\TUp{0} = \Set{1,2,3,4}$ such that
    $ \Set{4} = \TUp{1} \subsetneq  \TUp{0} = V_4 $
    and
    \begin{equation}
        \TU{\alpha}{V_4} = \begin{cases}
                             \Set{4} & \alpha \in [0,7],\\
                             \Set{1,\dotsc,4} & \alpha \in (7,10].
                         \end{cases} \nonumber
    \end{equation}
    We use $\TU{\alpha}{V_4}$ to update $\rv_{\alpha,V}$ and $\Qat{\alpha}{V_4}$ for all $\alpha$ as in step~\ref{step:UpdatesPar} and get
    \begin{equation} \label{eq:UpV4}
        \begin{aligned}
            & \rv_{\alpha,V_4} = \begin{cases}
                                    (\alpha-2, \alpha-4, \alpha-6, \alpha-6) & \alpha \in [0,4], \\
                                    (\alpha-2, 0, \alpha-6, \alpha-6) & \alpha \in (4,7], \\
                                    (\alpha-2, 0, \alpha-6, 8-\alpha) & \alpha \in (7,8], \\
                                    (\alpha-2,0,2,0) & \alpha \in (8,10],
                                \end{cases}\\
            & \Qat{\alpha}{V_4} = \begin{cases}
                                    \Set{\Set{1},\dotsc,\Set{4}} & \alpha \in [0,4], \\
                                    \Set{\Set{1,2},\Set{3},\Set{4}} & \alpha \in (4,7], \\
                                    \Set{\Set{1,\dotsc,4}} & \alpha \in (7,10].
                                \end{cases}
        \end{aligned}
    \end{equation}
    See Fig.~\ref{fig:PSP}(d).

    For $i = 5$, we have the critical points for the problem $\min\Set{ \FuU{\alpha}(\TX) \colon \Set{5} \in \X \subseteq \Qat{\alpha}{V} }$ being $\alpha_2 = 6$ , $\alpha_1 = 6.5$ and $\alpha_0 = H(V) = 10$ with $\TUp{2} = \Set{5}$, $\TUp{1} = \Set{1,2,5}$ and $\TUp{0} = \Set{1,\dotsc,5}$ such that
    \begin{equation}
        \TU{\alpha}{V} = \begin{cases}
                             \Set{5} & \alpha \in [0,6],\\
                             \Set{1,2,5} & \alpha \in (6,6.5], \\
                             \Set{1,\dotsc,5} & \alpha \in (6.5,10].
                         \end{cases} \nonumber
    \end{equation}
    After the updates in step~\ref{step:UpdatesPar}, we have
    \begin{equation} \label{eq:UpV5}
        \begin{aligned}
            & \rv_{\alpha,V} = \begin{cases}
                                    (\alpha-2, \alpha-4, \alpha-6, \alpha-6,\alpha-5) & \alpha \in [0,4] ,\\
                                    (\alpha-2, 0, \alpha-6, \alpha-6,\alpha-5) & \alpha \in (4,6] ,\\
                                    (\alpha-2, 0, \alpha-6, \alpha-6,1) & \alpha \in (6,6.5] ,\\
                                    (\alpha-2, 0, \alpha-6, \alpha-6,14-2\alpha) & \alpha \in (6.5,7] ,\\
                                    (\alpha-2, 0, \alpha-6, 8-\alpha,0) & \alpha \in (7,8] ,\\
                                    (\alpha-2,0,2,0,0) & \alpha \in (8,10] ,
                                \end{cases} \\
            & \Qat{\alpha}{V} = \begin{cases}
                                    \Set{\Set{1},\dotsc,\Set{5}} & \alpha \in [0,4], \\
                                    \Set{\Set{1,2},\Set{3},\Set{4},\Set{5}} & \alpha \in (4,6], \\
                                    \Set{\Set{1,2,5},\Set{3},\Set{4}} & \alpha \in (6,6.5], \\
                                    \Set{\Set{1,\dotsc,5}} & \alpha \in (6.5,10].
                                \end{cases}
        \end{aligned}
    \end{equation}
    See Fig.~\ref{fig:PSP}(e). Here, we finally update to $\Qat{\alpha}{V}$ for all $\alpha$. The corresponding PSP has the critical points $\alphap{3} = 4$, $\alphap{2} = 6$ and $\alphap{1} = 6.5$ and $\alphap{0} = H(V) = 10$ with $\Patp{3} = \Set{\Set{1},\dotsc,\Set{5}}$, $\Patp{2} = \Set{\Set{1,2},\Set{3},\Set{4},\Set{5}}$, $\Patp{1} = \Set{\Set{1,2,5},\Set{3},\Set{4}}$ and $\Patp{0} = \Set{\Set{1,\dotsc,5}}$ so that we know $\RCO(V) = \alphap{1} = 6.5$ is the minimum sum-rate and $\Patp{1} = \Set{\Set{1,2,5},\Set{3},\Set{4}}$ is finest maximizer of problem~\eqref{eq:MinSumRatePat}. We also know an optimal achievable rate vector $\rv_{6.5,V} = (4.5,0,0.5,0.5,1)$.
\end{example}

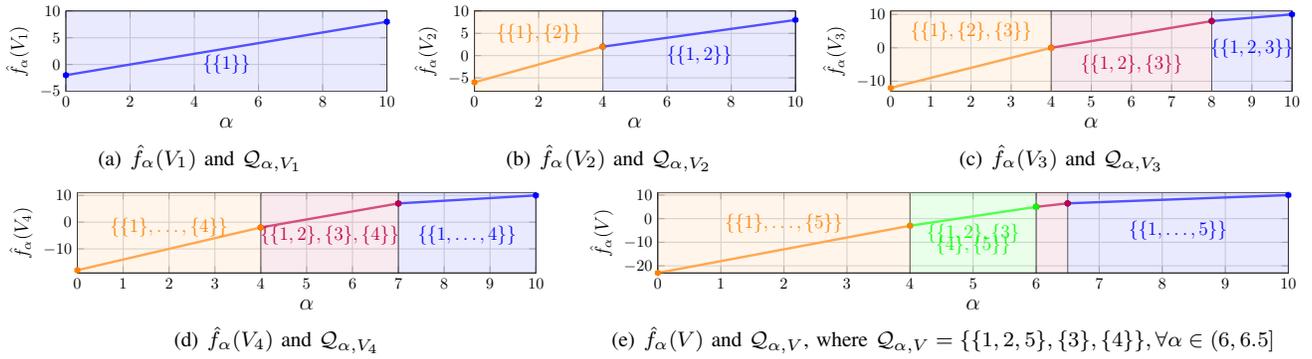
\begin{figure*}[t]
	\centering
    \subfigure[$\FuHat{\alpha}(V_1)$ and $\Qat{\alpha}{V_1}$]{\scalebox{0.6}{
%
%
\definecolor{mycolor1}{rgb}{1,0,1}%

\begin{tikzpicture}[
every pin/.style={fill=yellow!50!white,rectangle,rounded corners=3pt,font=\tiny},
every pin edge/.style={<-}]

\begin{axis}[%
width=2.8in,
height=0.7in,
scale only axis,
xmin=0,
xmax=10,
xlabel={\Large $\alpha$},
xmajorgrids,
ymin=-5,
ymax=10,
ylabel={\large $\FuHat{\alpha}(V_1)$},
ymajorgrids,
legend style={at={(0.65,0.05)},anchor=south west,draw=black,fill=white,legend cell align=left}
]

\addplot [
color=blue,
solid,
line width=1.5pt,
mark=asterisk,
mark options={solid}
]
table[row sep=crcr]{
10 8\\
0 -2\\
};

\addplot[area legend,solid,fill=blue!20,opacity=4.000000e-01]coordinates {
(10,10)
(10,-8)
(0,-8)
(0,10)
};
\node at (axis cs:5,0) {\textcolor{blue}{\large $\Set{\Set{1}}$}};


\end{axis}
\end{tikzpicture}%
}}
    \subfigure[$\FuHat{\alpha}(V_2)$ and $\Qat{\alpha}{V_2}$]{\scalebox{0.6}{
%
%
\definecolor{mycolor1}{rgb}{1,0,1}%

\begin{tikzpicture}[
every pin/.style={fill=yellow!50!white,rectangle,rounded corners=3pt,font=\tiny},
every pin edge/.style={<-}]

\begin{axis}[%
width=2.8in,
height=0.7in,
scale only axis,
xmin=0,
xmax=10,
xlabel={\Large $\alpha$},
xmajorgrids,
ymin=-8,
ymax=10,
ylabel={\large $\FuHat{\alpha}(V_2)$},
ymajorgrids,
legend style={at={(0.65,0.05)},anchor=south west,draw=black,fill=white,legend cell align=left}
]

\addplot [
color=blue,
solid,
line width=1.5pt,
mark=asterisk,
mark options={solid}
]
table[row sep=crcr]{
10 8\\
4 2\\
};

\addplot[area legend,solid,fill=blue!20,opacity=4.000000e-01]coordinates {
(10,10)
(10,-8)
(4,-8)
(4,10)
};
\node at (axis cs:7,0) {\textcolor{blue}{\large $\Set{\Set{1,2}}$}};

\addplot [
color=orange,
solid,
line width=1.5pt,
mark=asterisk,
mark options={solid}
]
table[row sep=crcr]{
4 2\\
0 -6 \\
};

\addplot[area legend,solid,fill=orange!20,opacity=4.000000e-01]coordinates {
(4,10)
(4,-8)
(0,-8)
(0,10)
};
\node at (axis cs:2,5) {\textcolor{orange}{\large $\Set{\Set{1},\Set{2}}$}};



\end{axis}
\end{tikzpicture}%
}}
    \subfigure[$\FuHat{\alpha}(V_3)$ and $\Qat{\alpha}{V_3}$]{\scalebox{0.6}{
%
%
\definecolor{mycolor1}{rgb}{1,0,1}%

\begin{tikzpicture}[
every pin/.style={fill=yellow!50!white,rectangle,rounded corners=3pt,font=\tiny},
every pin edge/.style={<-}]

\begin{axis}[%
width=3.5in,
height=0.7in,
scale only axis,
xmin=0,
xmax=10,
xlabel={\Large $\alpha$},
xmajorgrids,
ymin=-13,
ymax=11,
ylabel={\large $\FuHat{\alpha}(V_3)$},
ymajorgrids,
legend style={at={(0.65,0.05)},anchor=south west,draw=black,fill=white,legend cell align=left}
]

\addplot [
color=blue,
solid,
line width=1.5pt,
mark=asterisk,
mark options={solid}
]
table[row sep=crcr]{
10 10\\
8 8\\
};

\addplot[area legend,solid,fill=blue!20,opacity=4.000000e-01]coordinates {
(10,11)
(10,-13)
(8,-13)
(8,11)
};
\node at (axis cs:9,0) {\textcolor{blue}{\large $\Set{\Set{1,2,3}}$}};

\addplot [
color=purple,
solid,
line width=1.5pt,
mark=asterisk,
mark options={solid}
]
table[row sep=crcr]{
8 8\\
4 0\\
};

\addplot[area legend,solid,fill=purple!20,opacity=4.000000e-01]coordinates {
(8,11)
(8,-13)
(4,-13)
(4,11)
};
\node at (axis cs:6,-5) {\textcolor{purple}{\large $\Set{\Set{1,2},\Set{3}}$}};

\addplot [
color=orange,
solid,
line width=1.5pt,
mark=asterisk,
mark options={solid}
]
table[row sep=crcr]{
4 0\\
0 -12 \\
};

\addplot[area legend,solid,fill=orange!20,opacity=4.000000e-01]coordinates {
(4,11)
(4,-13)
(0,-13)
(0,11)
};
\node at (axis cs:2,5) {\textcolor{orange}{\large $\Set{\Set{1},\Set{2},\Set{3}}$}};



\end{axis}
\end{tikzpicture}%
}}
    \subfigure[$\FuHat{\alpha}(V_4)$ and $\Qat{\alpha}{V_4}$]{\scalebox{0.6}{
%
%
\definecolor{mycolor1}{rgb}{1,0,1}%

\begin{tikzpicture}[
every pin/.style={fill=yellow!50!white,rectangle,rounded corners=3pt,font=\tiny},
every pin edge/.style={<-}]

\begin{axis}[%
width=4in,
height=0.7in,
scale only axis,
xmin=0,
xmax=10,
xlabel={\Large $\alpha$},
xmajorgrids,
ymin=-19,
ymax=11,
ylabel={\large $\FuHat{\alpha}(V_4)$},
ymajorgrids,
legend style={at={(0.65,0.05)},anchor=south west,draw=black,fill=white,legend cell align=left}
]

\addplot [
color=blue,
solid,
line width=1.5pt,
mark=asterisk,
mark options={solid}
]
table[row sep=crcr]{
10 10\\
7 7\\
};

\addplot[area legend,solid,fill=blue!20,opacity=4.000000e-01]coordinates {
(10,11)
(10,-19)
(7,-19)
(7,11)
};
\node at (axis cs:8.5,-5) {\textcolor{blue}{\large $\Set{\Set{1,\dotsc,4}}$}};

\addplot [
color=purple,
solid,
line width=1.5pt,
mark=asterisk,
mark options={solid}
]
table[row sep=crcr]{
7 7\\
4 -2\\
};

\addplot[area legend,solid,fill=purple!20,opacity=4.000000e-01]coordinates {
(7,11)
(7,-19)
(4,-19)
(4,11)
};
\node at (axis cs:5.5,-5) {\textcolor{purple}{\large $\Set{\Set{1,2},\Set{3},\Set{4}}$}};

\addplot [
color=orange,
solid,
line width=1.5pt,
mark=asterisk,
mark options={solid}
]
table[row sep=crcr]{
4 -2\\
0 -18 \\
};

\addplot[area legend,solid,fill=orange!20,opacity=4.000000e-01]coordinates {
(4,11)
(4,-19)
(0,-19)
(0,11)
};
\node at (axis cs:2,-2) {\textcolor{orange}{\large $\Set{\Set{1},\dotsc,\Set{4}}$}};



\end{axis}
\end{tikzpicture}%
}} \quad
    \subfigure[$\FuHat{\alpha}(V)$ and $\Qat{\alpha}{V}$, where $\Qat{\alpha}{V} = \Set{\Set{1,2,5},\Set{3},\Set{4}},\forall \alpha \in {(6,6.5]}$]{\scalebox{0.6}{
%
%
\definecolor{mycolor1}{rgb}{1,0,1}%

\begin{tikzpicture}[
every pin/.style={fill=yellow!50!white,rectangle,rounded corners=3pt,font=\tiny},
every pin edge/.style={<-}]

\begin{axis}[%
width=5.5in,
height=0.7in,
scale only axis,
xmin=0,
xmax=10,
xlabel={\Large $\alpha$},
xmajorgrids,
ymin=-23,
ymax=11,
ylabel={\large $\FuHat{\alpha}(V)$},
ymajorgrids,
legend style={at={(0.65,0.05)},anchor=south west,draw=black,fill=white,legend cell align=left}
]

\addplot [
color=blue,
solid,
line width=1.5pt,
mark=asterisk,
mark options={solid}
]
table[row sep=crcr]{
10 10\\
6.5 6.5\\
};

\addplot[area legend,solid,fill=blue!20,opacity=4.000000e-01]coordinates {
(10,11)
(10,-23)
(6.5,-23)
(6.5,11)
};
\node at (axis cs:8.25,-5) {\textcolor{blue}{\large $\Set{\Set{1,\dotsc,5}}$}};

\addplot [
color=purple,
solid,
line width=1.5pt,
mark=asterisk,
mark options={solid}
]
table[row sep=crcr]{
6.5 6.5\\
6 5\\
};

\addplot[area legend,solid,fill=purple!20,opacity=4.000000e-01]coordinates {
(6.5,11)
(6.5,-23)
(6,-23)
(6,11)
};

\addplot [
color=green,
solid,
line width=1.5pt,
mark=asterisk,
mark options={solid}
]
table[row sep=crcr]{
6 5\\
4 -3 \\
};

\addplot[area legend,solid,fill=green!20,opacity=4.000000e-01]coordinates {
(6,11)
(6,-23)
(4,-23)
(4,11)
};
\node at (axis cs:5,-6) {\textcolor{green}{\large $\{\Set{1,2},\Set{3}$}};
\node at (axis cs:5,-11) {\textcolor{green}{\large $\Set{4},\Set{5} \}$}};

\addplot [
color=orange,
solid,
line width=1.5pt,
mark=asterisk,
mark options={solid}
]
table[row sep=crcr]{
4 -3\\
0 -23 \\
};

\addplot[area legend,solid,fill=orange!20,opacity=4.000000e-01]coordinates {
(4,11)
(4,-23)
(0,-23)
(0,11)
};
\node at (axis cs:2,-2) {\textcolor{orange}{\large $\Set{\Set{1},\dotsc,\Set{5}}$}};

\end{axis}
\end{tikzpicture}%
}}
	\caption{The piecewise linear strictly increasing Dilworth truncation $\FuHat{\alpha}(V_i)$ in $\alpha$ and the segmented partition $\Qat{\alpha}{V_i}$ obtained at the end of each iteration of the PAR Algorithm when it is applied to the system in Example~\ref{ex:main}.}
	\label{fig:PSP}
\end{figure*}

The rest of problem is how to obtain the critical points $\alpha_j$s and $\TUp{j}$s in Lemma~\ref{lemma:PP}. Since $\Qat{\alpha}{V_i}$ for all $\alpha$ determines all partitions in the PSP of $V_i$ (see Section~\ref{subsec:PSPExpand}), we can still use Lemma~\ref{lemma:AlphaAdapt} in Appendix~\ref{app:AlphaAdapt} to adapt the value of $\alpha$ so that all $\TUp{j}$s are determined by the call $\StrMap(\Set{\Set{m} \colon m \in V_i},\Set{V_i})$ in Algorithm~\ref{algo:PP}, and the corresponding $\alpha_j$s can be obtained by another property of the strict strong map property below.

\begin{lemma}[{\cite[Theorem~31]{Fujishige2009PP}}] \label{lemma:CrVals}
    For all $\alpha_j$s and $\TUp{j}$s that characterize $\TU{\alpha}{V_i}$ of the minimal minimizer of $\min\Set{ \FuU{\alpha}(\TX) \colon \Set{i} \in \X \subseteq \Qat{\alpha}{V_i}}$ in Lemma~\ref{lemma:PP},
    $$ r_{\alpha_j} ( \TUp{j-1} \setminus \TUp{j}) = H(\TUp{j-1}) - H(\TUp{j}), \ \forall j \in \Set{1,\dotsc,q}. \hfill\QED$$
\end{lemma}

\begin{example} \label{ex:StrMap}
    We call $\StrMap(\Set{\Set{1},\dotsc,\Set{5}}, \Set{\Set{1,\dotsc,5}})$ to determine all $\TUp{j}$s for $\min\Set{ \FuU{\alpha}(\TX) \colon \Set{5} \in \X \subseteq \Qat{\alpha}{V} }$ at the iteration $i =5$ in Example~\ref{ex:ParAlgo}. Here, $\Qat{\alpha}{V} = \Qat{\alpha}{V_4} \sqcup \Set{\Set{5}}$ where $\Qat{\alpha}{V_4}$ is in \eqref{eq:UpV4}. In this call, we have $\alpha = 5.75$ and $\SU = \bigcap \argmin \Set{ \FuU{5.75}(\TX) \colon \Set{5} \in \X \subseteq \Qat{5.75}{V}} = \Set{\Set{1,2},\Set{5}}$ so that $\Pat = (\Qat{5.75}{V} \setminus \SU) \sqcup \Set{ \TSU } = \Set{\Set{1,2,5},\Set{3},\Set{4}}$. Since $\Pat \neq \Set{\Set{1},\dotsc,\Set{5}}$, recursion continues.

    In the call $\StrMap(\Set{\Set{1},\dotsc,\Set{5}}, \Set{\Set{1,2,5},\Set{3},\Set{4}})$, we have $\alpha = 5$, $\SU = \Set{\Set{5}}$ and $\Pat = \Set{\Set{1,2},\Set{3},\Set{4},\Set{5}}$. We call $\StrMap(\Set{\Set{1},\dotsc,\Set{5}}, \Set{\Set{1,2},\Set{3},\Set{4},\Set{5}})$ and $\StrMap(\Set{\Set{1,2},\Set{3},\Set{4},\Set{5}}, \Set{\Set{1,2,5},\Set{3},\Set{4}})$ and get $\Set{5}$ returned for both calls.
    In the call $\StrMap(\Set{\Set{1,2,5},\Set{3},\Set{4}}, \Set{\Set{1,\dotsc,5}})$, we have $\alpha = 6.5$, $\SU = \Set{\Set{1,2},\Set{5}}$ and $\Pat = \Set{\Set{1,2,5},\Set{3},\Set{4}}$ so that it terminates with $\Set{1,2,5}$ returned.

    Finally we have $\Set{5}$ and $\Set{1,2,5}$ as $\TUp{2}$ and $\TUp{1}$, respectively, and know that $\TUp{0} = V = \Set{1,\dotsc,5}$. We apply Lemma~\ref{lemma:CrVals} to determine that $\alpha_2 = 6$, $\alpha_1 = 6.5$ and $\alpha_0 = H(V) = 10$.
\end{example}

      \begin{algorithm} [t]
	       \label{algo:PP}
	       \small
	       \SetAlgoLined
	       \SetKwInOut{Input}{input}\SetKwInOut{Output}{output}
	       \SetKwFor{For}{for}{do}{endfor}
            \SetKwRepeat{Repeat}{repeat}{until}
            \SetKwIF{If}{ElseIf}{Else}{if}{then}{else if}{else}{endif}
	       \BlankLine
           \Input{$\Pat_d,\Pat_u \in \Pi(V_i)$ such that $\Pat_d \prec \Pat_u$ (We assume the $\Qat{\alpha}{V_i}$ and $\FuU{\alpha}$ for all $\alpha$ are the global variables.)}
	       \Output{$\Set{\TUp{j} \colon j \in \Set{0,\dotsc,q}}$ for the problem $\min \Set{ \FuU{\alpha}(\TX) \colon \Set{i} \in \X \subseteq \Qat{\alpha}{V_i}}$}
	       \BlankLine
            $\alpha \leftarrow H(V) - \frac{ H[\Pat_d] - H[\Pat_u] }{ |\Pat_d| - |\Pat_u| }$\;
            $\SU \leftarrow \bigcap \argmin \Set{ \FuU{\alpha}(\TX) \colon \Set{i} \in \X \subseteq \Qat{\alpha}{V_i}} $\;
            $\Pat \leftarrow (\Qat{\alpha}{V_i} \setminus \SU) \sqcup \Set{ \TSU }$\;
            \lIf{$\Pat = \Pat_d$}{return $\Set{\TSU}$}
            \lElse{ return $\StrMap(\Pat_d,\Pat) \cup \StrMap(\Pat,\Pat_u)$ }
	   \caption{Strong Map (StrMap) Algorithm}
	   \end{algorithm}

\subsection{Complexity}
\label{sec:Complexity}

Due to the strong map property in Theorem~\ref{theo:StrongMap}, the $\StrMap$ algorithm can be completed by a parametric submodular function minimization (PSFM) algorithm, e.g., \cite{Fleischer2003PSFM,Nagano2007PSFM,IwataPSFM1997}, that runs in the same asymptotic time as a single call of a SFM algorithm. See Appendix~\ref{app:PSFM}.
Therefore, the minimum sum-rate problem can be solved by the PAR algorithm in $O(|V| \cdot \SFM(|V|))$ time. As compared to $O(|V|^2 \cdot \SFM(|V|))$ time of the MDA algorithm in \cite{Ding2018IT}, the complexity is reduced by a factor of $|V|$.

\subsubsection{CO in Expanding Ground Set}
\label{subsec:PSPExpand}

An observation on Fig.~\ref{fig:PSP}, the minimum sum-rate problem can be solved when the size of the ground set $V_i$ is gradually increasing in the order of $i = 1,2,\dotsc,|V|$.\footnote{This is particularly useful when the users complete recording their observations in different times. For example $i = 1$ is assigned to the user that completes observations first.}
In addition, by replacing the horizontal axis by $\alpha \leftarrow \alpha - H(V) + H(V_i)$ in Figs.~\ref{fig:PSP}(a)-(e), we have them being exactly the PSP of $V_i$ that provides the solution to the minimum sum-rate problem in $V_i$. For example, in Fig.~\ref{fig:PSP}(b), by letting $\alpha \leftarrow \alpha - H(V) + H(V_2) = \alpha - 2$, we have $\Qat{\alpha}{V_2}$ determining the PSP of $V_2$ so that the first critical point $\alphap{1} = 2 = \RCO(V_2)$. See Fig.~\ref{fig:PSPExp}.

\subsubsection{Distributed Implementation}
\label{subsec:Distr}

The PAR algorithm can be implemented in a distributed manner: for all $i = 1,\dotsc,|V|$, let user $i$ obtain $\Qat{\alpha}{V_i}$ and $\rv_{\alpha,V_i}$ for all $\alpha$ and pass them to user $i+1$.

\begin{figure}[t]
	\centering
    \scalebox{0.6}{
%
%
\definecolor{mycolor1}{rgb}{1,0,1}%

\begin{tikzpicture}[
every pin/.style={fill=yellow!50!white,rectangle,rounded corners=3pt,font=\tiny},
every pin edge/.style={<-}]

\begin{axis}[%
width=3.8in,
height=0.7in,
scale only axis,
xmin=0,
xmax=8,
xlabel={\Large $\alpha$},
xmajorgrids,
ymin=-4,
ymax=10,
ylabel={\large $\FuHat{\alpha}(V_2)$},
ymajorgrids,
legend style={at={(0.65,0.05)},anchor=south west,draw=black,fill=white,legend cell align=left}
]

\addplot [
color=blue,
solid,
line width=1.5pt,
mark=asterisk,
mark options={solid}
]
table[row sep=crcr]{
8 8\\
2 2\\
};

\addplot[area legend,solid,fill=blue!20,opacity=4.000000e-01]coordinates {
(8,10)
(8,-8)
(2,-8)
(2,10)
};
\node at (axis cs:5,0) {\textcolor{blue}{\large $\Set{\Set{1,2}}$}};

\addplot [
color=orange,
solid,
line width=1.5pt,
mark=asterisk,
mark options={solid}
]
table[row sep=crcr]{
2 2\\
0 -2 \\
};

\addplot[area legend,solid,fill=orange!20,opacity=4.000000e-01]coordinates {
(2,10)
(2,-8)
(0,-8)
(0,10)
};
\node at (axis cs:1,5) {\textcolor{orange}{\large $\Set{\Set{1},\Set{2}}$}};

\end{axis}
\end{tikzpicture}%
}
	\caption{The $\FuHat{\alpha}(V_2)$ vs, $\alpha$ plot in Fig.~\ref{fig:PSP}(b) by a shift $\alpha \leftarrow \alpha - H(V) + H(V_2)$. By doing so, the plot characterizes the PSP of $V_2$ with critical points $\alphap{1} = 2$ and $\alphap{0} = H(V_2) = 8$. Here, $\alphap{1} = \RCO(V_2)$ is the minimum sum-rate for attaining the omniscience in $V_2$ with an optimal rate vector $\rv_{2,\Set{1,2}} = (2,0)$.}
	\label{fig:PSPExp}
\end{figure}
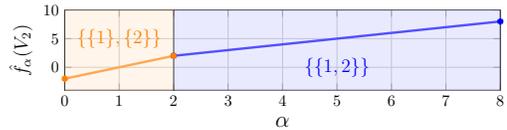

\section{Successive Omniscience}
\label{sec:SO}

Successive omniscience (SO) is proposed in \cite{ChanSuccessiveIT,Ding2015NetCod}. The idea is based on the fact that there is a particular group of nonsingleton subsets $\XComp$s of $V$ such that the local omniscience in $\XComp$ can be attained with the minimum sum-rate $\RCO(\XComp)$ before the global omniscience in $V$ so that the overall sum-rate for the CO problem still remains the minimum $\RCO(V)$. We call $\XComp$ a \textit{complimentary subset} (for SO) \cite{ChanSuccessiveIT}. The existence of a complimentary subset suggests that omniscience can be attained in a successive manner: first select and attain local omniscience in $\XComp$; then solve the global omniscience problem in $V$. An example in \cite[Section~IV]{Ding2017ISIT} shows that recursively applying this procedure leads to a multi-stage SO towards the global omniscience.

Obviously, the main task in SO is to select a subset $\XComp$ that is complimentary. It is shown in \cite{ChanSuccessiveIT} that a non-singleton $\XComp \subsetneq V$ is complimentary if and only if $H(V) - H(\XComp) + \RCO(\XComp) \leq \RCO(V)$. The interpretation is that: after local omniscience in $\XComp$ is attained by minimum sum-rate $\RCO(\XComp)$, the other users in $V \setminus \XComp$ should transmit at least the missing randomness $H(V) - H(\XComp)$ to attain global omniscience; the overall rate $H(V) - H(\XComp) + \RCO(\XComp)$ must be no greater than $\RCO(V)$ for $\XComp$ to be complimentary. This necessary and sufficient condition is shown to be equivalent to $\Fu{\RCO(V)} (\XComp) = \FuHat{\RCO(V)} (\XComp)$ \cite[Corollary~III.3]{Ding2017ISIT}.
The authors in \cite{Ding2017ISIT} further relaxed it to a sufficient condition based on a lower bound on $\RCO(V)$ as follows.

\begin{proposition}[{\cite[Lemma~III.7]{Ding2017ISIT}}] \label{prop:SO}
    A nonsingleton $\XComp \subsetneq V$ is complimentary if $\Fu{\alphaU}(\XComp) = \FuHat{\alphaU}(\XComp)$ for $\alphaU \leq \RCO(V)$. \hfill \QED
\end{proposition}

However, to implement SO, we still need to know an optimal rate vector for attaining the local omniscience in $\XComp$ with sum-rate $\RCO(\XComp)$, which is not addressed in \cite{Ding2018IT} or the existing literature. One may think of a two step approach: choose a $\XComp$; solve the minimum sum-rate problem in $\XComp$ by the existing CO techniques, e.g., the MDA algorithm in \cite{Ding2018IT}.
But, it is worth discussing if there exists a less complex approach: can an optimal rate vector for the local omniscience in $\XComp$ be obtained at the same time when $\XComp$ is selected? The following theorem shows that we can do so by utilizing the results in the PAR algorithm. The proof is in Appendix~\ref{app:theo:SO}.

\begin{theorem} \label{theo:SO}
    Let $\Qat{\alpha}{V_i}$ and $\rv_{\alpha,V_i}$ be segmented partition and rate vector, respectively, obtained at the end of iteration $i \in \Set{2,\dotsc,|V|}$ of the PAR algorithm. For $\alphaU = \sum_{i \in V} \frac{H(V) - H(\Set{i})}{|V| - 1}$, all $C \in \Qat{\alphaU}{V_i}$ such that $|C| > 1$ satisfy $\Fu{\alphaU} (C) = \FuHat{\alphaU} (C)$, i.e., they are complimentary subsets. Also, let $\alphaComp$ be the minimum value of $\alpha$ such that $\Fu{\alpha} (C) = \FuHat{\alpha} (C)$. $\rv_{\alphaComp,C}$ is an optimal rate vector that attains local omniscience in $C$ with the minimum sum-rate $\RCO(C)$. \hfill \QED
\end{theorem}

Theorem~\ref{theo:SO} suggests that a complimentary subset for SO, if there exists any, can be searched at the end of each iteration of the PAR algorithm. Therefore, we have the SO algorithm in Algorithm~\ref{algo:SO} that completes in $O(|V| \cdot \SFM(|V|))$ time.
In addition, $\alphaComp \leq \alphaU$ always and $\alphaComp$ is not hard to determine from $\Qat{\alpha}{V_i}$: It is the critical value or turning point where all subsets of $C$ merges to $C$. See the example below.

        \begin{algorithm} [t]
	       \label{algo:SO}
	       \small
	       \SetAlgoLined
	       \SetKwInOut{Input}{input}\SetKwInOut{Output}{output}
	       \SetKwFor{For}{for}{do}{endfor}
            \SetKwRepeat{Repeat}{repeat}{until}
            \SetKwIF{If}{ElseIf}{Else}{if}{then}{else if}{else}{endif}
	       \BlankLine
           \Input{$V$ and $H$}
	       \Output{a complimentary subset $C$ and an optimal rate vector $\rv_{\alphaComp,C}$ for attaining the local omniscience in $C$}
	       \BlankLine
           $\alphaU \leftarrow \sum_{i \in V} \frac{H(V) - H(\Set{i})}{|V| - 1}$\;
           \For{$i=2$ \emph{\KwTo} $|V|$}{
                Let $\Qat{\alpha}{V_i}$ and $\rv_{\alpha,V_i}$ be the segmented variables obtained at the end of the $i$th iteration in the PAR algorithm\;
                \lIf{$\exists C \in \Qat{\alphaU}{V_i} \colon |C| > 1$}{
                    return $C$ and $\rv_{\alphaComp,C}$ where $\alphaComp = \min \Set{ \alpha \colon \Fu{\alpha} (C) = \FuHat{\alpha} (C) }$
                }
            }
	   \caption{Successive Omniscience (SO) Algorithm}
	   \end{algorithm}

\begin{example} \label{ex:SO}
    We apply Theorem~\ref{theo:SO} to the system in Example~\ref{ex:main}. We set $\alphaU = \sum_{i \in V} \frac{H(V) - H(\Set{i})}{|V| - 1} = 5.75$. For $i = 2$, we obtain the $\Qat{\alpha}{V_2}$ and $\rv_{\alpha,V_2}$ for all $\alpha$ as in \eqref{eq:ExQatR} at the end of the iteration $i = 2$ in the PAR algorithm so that $\Qat{\alphaU}{V_2} = \Set{\Set{1,2}}$. The only nonsingleton $\Set{1,2}$ is a complimentary subset. We then search region $[0, \alphaU]$ and find that $\Set{1}$ and $\Set{2}$ merge to $\Set{1,2}$ at $\alphaComp = 4$. See Fig.~\ref{fig:SO2}. Note, this meas $\alphaComp = \min \Set{ \alpha \colon \Fu{\alpha} (\Set{1,2}) = \FuHat{\alpha} (\Set{1,2}) } = 4$, where we have $\rv_{4,V_2} = (2,0)$ being an optimal rate vector that attains local omniscience in $\Set{1,2}$ with $\RCO(\Set{1,2}) = r_{4}(\Set{1,2}) = 2$. For the optimal rate vector $(4.5,0,0.5,0.5,1)$ for the global omniscience obtained in Example~\ref{ex:ParAlgo}, we have $(4.5,0,0.5,0.5,1) = (2,0,0,0,0) + (2.5,0,0.5,0.5,1)$, which means that, by letting users transmit at rate $(2.5,0,0.5,0.5,1)$ after the local omniscience in $\Set{1,2}$, the global omniscience is still attained with the minimum sum-rate $\RCO(V)$.

    Also note that Theorem~\ref{theo:SO} in fact applies to any lower bound $\alphaU \leq \RCO(V)$. For example, for $\alphaU = 6.25$, consider $\Qat{\alpha}{V}$ obtained by the PAR algorithm for $i = 5$ in Fig.~\ref{fig:PSP}(e). We have $\Set{1,2,5} \in \Qat{6.25}{V}$ being another complimentary subset with $\alphaComp = \min \Set{ \alpha \colon \Fu{\alpha} (\Set{1,2,5}) = \FuHat{\alpha} (\Set{1,2,5}) } = 6$ and $\rv_{6,\Set{1,2,5}} = (4,0,1)$ being an optimal rate vector for attaining the local omniscience in $\Set{1,2,5}$.
\end{example}

\begin{figure}[t]
	\centering
    \scalebox{0.6}{
%
%
\definecolor{mycolor1}{rgb}{1,0,1}%

\begin{tikzpicture}[
every pin/.style={rectangle,rounded corners=3pt,font=\tiny},
every pin edge/.style={<-}]

\begin{axis}[%
width=3.8in,
height=0.7in,
scale only axis,
xmin=0,
xmax=10,
xlabel={\Large $\alpha$},
xmajorgrids,
ymin=-8,
ymax=10,
ylabel={\large $\FuHat{\alpha}(V_2)$},
ymajorgrids,
legend style={at={(0.65,0.05)},anchor=south west,draw=black,fill=white,legend cell align=left}
]

\addplot [
color=blue,
solid,
line width=1.5pt,
mark=asterisk,
mark options={solid}
]
table[row sep=crcr]{
10 8\\
4 2\\
};

\addplot[area legend,solid,fill=blue!20,opacity=4.000000e-01]coordinates {
(10,10)
(10,-8)
(4,-8)
(4,10)
};
\node at (axis cs:7,0) {\textcolor{blue}{\large $\Set{\Set{1,2}}$}};

\addplot [
color=orange,
solid,
line width=1.5pt,
mark=asterisk,
mark options={solid}
]
table[row sep=crcr]{
4 2\\
0 -6 \\
};

\addplot[area legend,solid,fill=orange!20,opacity=4.000000e-01]coordinates {
(4,10)
(4,-8)
(0,-8)
(0,10)
};
\node at (axis cs:2,5) {\textcolor{orange}{\large $\Set{\Set{1},\Set{2}}$}};

\addplot [
color=red,
dashed,
line width=2pt,
]
table[row sep=crcr]{
5.75 10\\
5.75 -8\\
};
\node at (axis cs:5.75,-4.5) [pin={[pin distance = 5mm,pin edge = {red}]0:\textcolor{red}{\Large $\alphaU$}}] {};

\addplot[color=orange,dotted,line width=3pt] coordinates {
(4,10)
(4,-8)
};
\node at (axis cs:4,-4.5) [pin={[pin distance = 5mm,pin edge = {orange}]180:\textcolor{orange}{\Large $\alphaComp$}}] {};

\end{axis}
\end{tikzpicture}%
}
	\caption{For the the system in Example~\ref{ex:main}, when Theorem~\ref{theo:SO} is applied to $\Qat{\alpha}{V_2}$ in Fig.~\ref{fig:PSP}(b), we have $\alphaU = 5.75$ and $\Set{1,2} \in \Qat{5.75}{V_2}$ is a compliment subset for the SO. We obtain $\alphaComp = 4$ where $\Set{1}$ and $\Set{2}$ are merged to form a whole subset $\Set{1,2}$. For the $\rv_{\alpha,V_2}$ in \eqref{eq:ExQatR}, $\rv_{4,V_2} = (2,0)$ is an optimal rate vector for the local omniscience in $\Set{1,2}$. }
	\label{fig:SO2}
\end{figure}
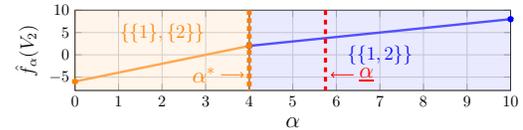

\section{Conclusion}

We reduced the complexity of solving the minimum sum-rate problem in CO from $O(|V|^2 \cdot \SFM(|V|))$ to $O(|V| \cdot \SFM(|V|))$ by proposing a PAR algorithm. We proved the strict strong map property, which ensures that the partitions and the rate vectors obtained in the existing $\CoordSatCapFus$ algorithm are segmented in $\alpha$, the estimation of the minimum sum-rate. We proposed the PAR algorithm where the critical points for the segmented variables can be searched by the PSFM algorithm so that the overall complexity reduces to $O(|V| \cdot \SFM(|V|))$. We discussed how to apply the PAR algorithm to a growing ground set in a distributed manner.
We also showed how to determine a complimentary user subset for SO and an optimal rate vector for attaining the local omniscience in it in $O(|V| \cdot \SFM(|V|))$ time.

There is a potential that the PAR algorithm can be adapted so that the PSP is obtained in a growing ground set $V_i$ in a fully distributed manner. Note, in the distributed implementation in Section~\ref{subsec:Distr}, the users still need to know the overall information $H(V)$ of the multiple source. Also, there should be more discussion on how to apply the PAR algorithm in a multi-stage SO.

\appendices

\section{} 
\label{app:AlphaAdapt}

Let $H[\Pat] = \sum_{C \in \Pat} H(C)$ for all $\Pat \in \Pi(V)$. The MDA algorithm proposed in \cite{Ding2018IT} starts with $\Pat = \Set{\Set{i} \colon i \in V}$ and run the recursion $(\rv_V,\Pat) \coloneqq \CoordSatCapFus(\alpha,V,H)$ where $\alpha = H(V) - \frac{ H[\Pat] - H(V) }{ |\Pat| - 1 }$ until $\alpha$ converges to $\RACO(V)$. The validity of the MDA algorithm is based on the properties of the PSP below. Note, Lemma~\ref{lemma:AlphaAdapt} also ensures the validity of $\StrMap$ algorithm in Algorithm~\ref{algo:PP}.

\begin{lemma}[{\cite[Sections~2.2 and 3]{MinAveCost}\cite[Definition~3.8]{Narayanan1991PLP}}] \label{lemma:AlphaAdapt}
    The $\alphap{j}$s and $\Patp{j}$s in the PSP of the ground set $V$ satisfy the followings.
    \begin{itemize}
      \item For all $j$, $\alphap{j} = H(V) - \frac{ H[\Patp{j}] - H[\Patp{j-1}] }{ |\Patp{j}| - |\Patp{j-1}| }$;
      \item For all $j,j'$ such that $j+1 < j'$, let $\alpha = H(V) - \frac{ H[\Patp{j'}] - H[\Patp{j}] }{ |\Patp{j'}| - |\Patp{j}| }$. Then, $\alphap{j} < \alpha \leq \alphap{j'}$
    \end{itemize}
\end{lemma}

\section{}
\label{app:PSFM}


In \cite{ParMaxFlow1989}, the push-relabel max-flow algorithm \cite{MaxFlow1988} was extended to the one that is parameterized by a parameter $\alpha$.
The same technique was further applied to extend the SFM algorithms to the PSFM ones in \cite{Fleischer2003PSFM,Nagano2007PSFM,IwataPSFM1997}. For example, the PSFM algorithm proposed in \cite{Fleischer2003PSFM} nests the Schrijver's SFM algorithm \cite{Schrijver2000SFM} in a push-relabel framework so that, for the function $\FuU{\alpha}$ that forms strong map sequence in $\alpha$, all minimizers of $\min \Set{ \FuU{\alpha}(\TX) \colon \Set{i} \in \X \subseteq \Qat{\alpha}{V_i}}$ of decreasing or increasing values of $\alpha$ can be determined in the same asymptotic time as the Schrijver's algorithm. See \cite[Section~4.2]{Fleischer2003PSFM} for the details. 

\section{}
\label{app:theo:SO}

\begin{proof}
    In Theorem~\ref{theo:SO}, we have $\alphaU = \sum_{i \in V} \frac{H(V) - H(\Set{i})}{|V| - 1} \leq \RCO(V)$ based on \eqref{eq:MinSumRatePat}.
    We prove $\Fu{\alphaU} (C) = \FuHat{\alphaU} (C)$ for all $C \in \Qat{\alphaU}{V_i}$ such that $|C| > 1$ by contradiction. Assume that $\Fu{\alphaU} (C) \neq \FuHat{\alphaU} (C)$. Then, there must exist some $\Pat \in \Pi(C) \setminus \Set{C}$ such that $\Fu{\alphaU} (C) > \FuHat{\alphaU}[\Pat]$ so that $\Fu{\alphaU}[\Qat{\alphaU}{V_i}] > \Fu{\alphaU}[\Pat] + \sum_{C' \in \Qat{\alphaU}{V_i} \colon C' \neq C} \Fu{\alpha} (C')$. This contradicts $\Qat{\alphaU}{V_i} \in \argmin_{\Pat \in \Pi(V_i)} \Fu{\alphaU} [\Pat] $ in Proposition~\ref{prop:preamble}. Therefore, we must have $\Fu{\alphaU} (C) = \FuHat{\alphaU} (C)$. So, based on Proposition~\ref{prop:SO}, all nonsingleton $C \in \Qat{\alphaU}{V_i}$ are complimentary.

    For $\alphaComp = \min \Set{ \alpha \colon \Fu{\alpha} (C) = \FuHat{\alpha} (C) }$, we have $\alphaComp = H(V) - H(C) + \max_{\Pat \in \Pi(C) \colon |\Pat| > 1} \sum_{X \in \Pat} \frac{H(C) - H(X)}{|\Pat| - 1} = H(V) - H(C) + \RCO(C)$. Also, $\rv_{\alphaComp,C} \in B(\FuHat{\alphaComp}^{C},\leq)$ where $B(\FuHat{\alphaComp}^{C},\leq) = \Set{\rv_{\alphaComp,C} \in P(\Fu{\alphaComp}^{C},\leq) \colon r_{\alphaComp} (C) = \FuHat{\alphaComp} (C) = \Fu{\alphaComp} (C)}$.\footnote{Here, we use the property $P(\Fu{\alphaComp}^{C},\leq) = P(\FuHat{\alphaComp}^{C},\leq)$ \cite[Theorems 2.5(i) and 2.6(i)]{Fujishige2005}.}
    We have $r_{\alphaComp}(X) \leq \Fu{\alphaComp}(X) = H(X) + \alphaComp - H(V) = H(X) + \RCO(C) - H(C)$ for all $X \subseteq C$ in $P(\Fu{\alphaComp}^{C},\leq)$. These constraints ensure that $\rv_{\alphaComp,C} \in \RRCO(C)$. We also have $r_{\alphaComp} (C) = \Fu{\alphaComp} (C) = H(C) + \alphaComp - H(V) = \RCO(C)$. Therefore, $\rv_{\alphaComp,C}$ is an optimal rate vector for the local omniscience in $C$.
\end{proof}


\bibliographystyle{IEEEtran}
\bibliography{COSO(PSFM)BIB}

\end{document}